%% file: sn-article.tex
\documentclass[pdflatex,sn-mathphys-num]{sn-jnl}

\usepackage{graphicx}%
\usepackage{multirow}%
\usepackage{amsmath,amssymb,amsfonts}%
\usepackage{amsthm}%
\usepackage{mathrsfs}%
\usepackage[title]{appendix}%
\usepackage{xcolor}%
\usepackage{textcomp}%
\usepackage{manyfoot}%
\usepackage{booktabs}%
\usepackage{algorithm}%
\usepackage{algorithmicx}%
\usepackage{algpseudocode}%
\usepackage{listings}%
\usepackage{tikz}
\usepackage{forest}
\usepackage{pgfplots}
\pgfplotsset{compat=1.18}

\theoremstyle{thmstyleone}%
\newtheorem{theorem}{Theorem}%

\theoremstyle{thmstyletwo}%
\newtheorem{example}{Example}%
\newtheorem{remark}{Remark}%
\newtheorem{notation}{Notation}

\theoremstyle{thmstylethree}%
\newtheorem{definition}{Definition}%
\newtheorem{prule}{Rule}
\newtheorem{corollary}{Corollary}
\newtheorem{lemma}{Lemma}

\usepackage{scalerel}

\newcommand{\bigcircop}{
  \mathop{\scalerel*{\circ}{\sum}}\limits}

\begin{document}

\title[Article Title]{An \(\mathcal{O}(n)\) Space Construction of Superpermutations}

\author[1]{\fnm{Dhruv} \sur{Ajmera}}\email{dhruvajmera0@gmail.com}

\affil[1]{ \orgname{Lone Star High School}, \orgaddress{\street{Panther Creek Pkwy}, \city{Frisco}, \postcode{75033}, \state{Texas}, \country{United States}}}

\abstract{A superpermutation is a sequence that contains every permutation of \(n\) distinct symbols as a contiguous substring. For instance, a valid example for three symbols is a sequence that contains all six permutations. This paper introduces a new algorithm that constructs such sequences more efficiently than existing recursive and graph-theoretic methods. Unlike traditional techniques that suffer from scalability and factorial memory demands, the proposed approach builds superpermutations directly and compactly. This improves memory usage, enabling the construction of larger sequences previously considered impractical.}


\maketitle

\section{Introduction}\label{introduction}
The combinatorial structure of superpermutations—strings that embed all \(n!\) permutations of a symbol set as contiguous substrings—has fascinated researchers due to its connections with encoding theory, sequence compression, and Hamiltonian path constructions. Traditional investigations focus on minimizing the length of such sequences, a goal that has produced both theoretical bounds and computational challenges.

This paper investigates a less-explored but practically important aspect of these sequences: optimizing the algorithm utilized to generate them. Existing methods for constructing superpermutations—typically recursive or graph-theoretic—require cumulative construction of subcases \citep{bib1,bib4} or the discovery of shortest Hamiltonian paths \citep{bib1,bib2,bib5}. These strategies suffer from factorial time complexity coupled with factorial (or worse) space complexity, making them infeasible at even minor scales. In contrast, our method constructs superpermutations through a deterministic chaining of \hyperref[b]{beads}. It doesn't require the construction of subcases and incrementally outputs the result, avoiding storing the entire sequence in memory because of pre-calculated overlap lengths (a feat that is not feasible with recursive methodologies). The algorithm achieves a worst-case runtime of \(\mathcal{O}(n!)\) with only \(\mathcal{O}(n)\) space usage. Details of the implementation, theoretical properties, and comparative benchmarks follow in subsequent sections.

\section{Results}\label{results}

We analyze the performance of our algorithm through asymptotic comparison with both recursive and graph-theoretic approaches to constructing superpermutations. Theoretical runtime and space complexity expressions for each method are derived in Section~\ref{tsc}. While recursive methods operate in \(\mathcal{O}(n!)\) time, graph-theoretic approaches are significantly worse, exhibiting time complexity on the order of \(\mathcal{O}(2^{n!}n!^2)\) due to the combinatorial explosion in possible paths through the graph of permutations. Both methods also suffer from high space complexity: \(\mathcal{O}(n!)\) for recursive methods and \(\mathcal{O}(n!^2)\) for graph-theoretic techniques. In contrast, our algorithm's incremental construction allows it to run in \(\mathcal{O}(n!)\) time using only \(\mathcal{O}(n)\) space, as it maintains only the current working \hyperref[b]{bead} and a minimal buffer. 

To contextualize these differences, Table~\ref{tab:complexity} summarizes the high-level asymptotic behavior of each approach. A more detailed breakdown and derivation of these time and space complexities is outlined in Section~\ref{tsc}.

\begin{table}[h]
\centering
\caption{Asymptotic time and space complexity of three construction methods.}
\label{tab:complexity}
\begin{tabular}{@{}lll@{}}
\toprule
\textbf{Method} & \textbf{Time Complexity} & \textbf{Space Complexity} \\ \midrule
\textbf{This Work} & \(\mathcal{O}(n!)\) & \(\mathcal{O}(n)\) \\
Recursive & \(\mathcal{O}(n!)\) & \(\mathcal{O}(n!)\) \\
Graph-Theoretic & \(\mathcal{O}(2^{n!}n!^2)\) & \(\mathcal{O}(n!^2)\) \\ \bottomrule
\end{tabular}
\end{table}

\section{Framework}\label{framework}
\subsection{Beads}\label{beads}

Firstly, we fix an integer \(n\ge3\). Then, our alphabet is the ordered tuple \(s = (c_1,c_2,\dots,c_n)\), with each $c_k$ being a \emph{letter} in the alphabet.

Next, let \(\langle\cdot\rangle:s^{*} \longmapsto \Sigma^{*}\) be the concatenation map that sends a finite tuple \((x_1,\dots,x_k)\) of letters to the string \(x_1x_2\dots x_k\), where \(s^{*}\) is the set of all finite tuples over \(s\) and \(\Sigma^{*}\) is the set of strings.

Henceforth, we identify a tuple with its flattened string without further comment. Additionally, we identify tuple entries $x_i$ with their corresponding letters in the flattened string.

Now, we begin by introducing the core building block of our construction: the \textit{bead}. Beads serve as compact containers for multiple permutations and are designed to maximize space efficiency while maintaining structural simplicity. This allows us to build superpermutations with predictable overlap behavior.

\begin{definition}[Bead]\label{b}\normalfont
A \emph{bead} is a $(2n-1)$‑tuple
\[
  b=(x_1,\dots,x_{2n-1})
\]
such that every length‑\(n\) window
\((x_j,\dots,x_{j+n-1})\), \(1\le j\le n\),
\begin{enumerate}
  \item contains each letter \(c_1,\dots,c_n\) exactly once (i.e.\ is a permutation), and
  \item the \(n\) windows are pairwise distinct.
\end{enumerate}
We denote the set of all such tuples as \(\mathcal{R}^0\).
\end{definition}

Now, we introduce a formal notion of what space efficiency in encoding permutations is and how to quantify it.

\begin{definition}[Space Efficiency]\normalfont

    Space efficiency is the ratio of the length of a string to the number of permutations it contains. Lower ratios indicate greater space efficiency.
    
\end{definition}

We now establish the fundamental properties of beads. These results demonstrate that beads are optimally compact and that they exhibit predictable behavior, properties that we will later exploit when constructing larger structures. Part (i) of this proof motivates our usage of beads as the base structure for construction.

\begin{theorem}[Bead Properties]\label{bp} \normalfont 

The following properties hold for all beads.
    \begin{itemize}
        \item[(i)] A bead is the most space-efficient structure for containing permutations.
        \item[(ii)] Exactly \((n-1)!\) distinct beads are needed in a superpermutation for \(n\).
        \item[(iii)] Every \(m\)th letter of a bead, such that \(n < m \le 2n-1\), is a copy of the \(m-n\)th letter.
    \end{itemize}
    
\end{theorem}

\begin{proof}
We prove each statement individually.

\medskip
\noindent\textbf{(i)} \textit{A bead is the most space-efficient structure for containing permutations.}

\noindent Consider \(b_1\) of length \(2n - 1\). Appending \(x_k\) (where \(k \le n\)) duplicates an existing permutation, so no additional permutations can be added in this manner. 

Furthermore, a string of length $\ell$ has at most $\ell-n+1$ distinct length‑$n$ windows (the window can begin only at positions $1,\dots,\ell-n+1$), so $\ell-n+1\ge n$ and therefore $\ell\ge 2n-1$.  Because a bead has a length of exactly $2n-1$, no shorter string can realize the required windows, making the bead length‑minimal and thus space‑optimal.

\medskip
\noindent\textbf{(ii)} \textit{We require \((n - 1)!\) distinct beads in a superpermutation for \(n\).}

\noindent Each bead contains \(n\) distinct permutations. Since we require \(n!\) distinct permutations, the number of beads required is \(\frac{n!}{n} = (n - 1)!\).

\medskip
\noindent\textbf{(iii)} \textit{Every \(m\)th letter of a bead, such that \(n < m \le 2n - 1\), is a copy of the \((m - n)\)th letter.}

\noindent By Definition \ref{b} the windows \(x_1\!\dots x_n\) and \(x_{m-n+1}\!\dots x_m\) must differ. This implies that for any index \(m > n\), \(x_m = x_{m-n}\).
\end{proof}

From part (iii) of the previous proof, we will now notate any bead \(b\) as 
\[
b = (x_1,\dots ,x_n,x_1,\dots, x_{n-1}),
\]
due to the symmetry of letters about either side of the $n$-th letter.

Thus, we use beads as our base case due to their inherent compactness in construction.
\subsection{Intersections}\label{intersections}

To efficiently construct superpermutations, we must transition from one bead to the next while utilizing as much of the previous structure as possible. This is achieved via controlled overlaps between beads, which we formalize as intersections. 

\begin{definition}[Intersection] \normalfont 

An intersection of length \(\ell\) is a common substring that is shared between two structures.

\end{definition}

The following definition will be for ease of reference throughout future proofs.

\begin{definition}[Leading letter] \normalfont 

The leading letter of some sequence is the letter that occupies the first index of that sequence.

\end{definition}

Now, we define a function that enables us to transition between beads, as will be proven later. Thus, we are able to establish the first steps of generating the superpermutation from successive beads.

\begin{definition}[Straight‑shift $\mathrm{SS}_{n-k}$]\label{ss} \normalfont
    Let $1\le k\le n-2$ and $b=(x_1,\dots,x_{2n-1})\in\mathcal R^0$.
    Define $y=\mathrm{SS}_{n-k}(b)$ as:
    \[
    y_m=
    \begin{cases}
    x_{m+1}, & 1\le m<n-k,\\[2pt]
    x_1,     & m=n-k,\\[2pt]
    x_m,     & n-k<m\le n,\\[2pt]
    y_{m-n}, & n<m\le 2n-1.
    \end{cases}
    \]
\end{definition}

An example of the straight-shift function at index \(n-2\) is shown below for the set \{$c_1, c_2, \dots, c_5$\}.

\begin{figure}[htbp]
    \centering
    \begin{tikzpicture}
        \node (b1) at (0,0) {\texttt{$c_1c_2c_3c_4c_5c_1c_2c_3c_4$}};
        \draw[->, thick] (b1.east) -- ++(1,0);
        \node (b2) at (4.2,0) {\texttt{$c_2c_3c_1c_4c_5c_2c_3c_1c_4$}};
    \end{tikzpicture}
    \caption{Application of straight-shift \(\mathrm{SS}_{n-2}\) on a bead over \(\{c_1, c_2, \dots, c_5\}\), where the leading letter is re-inserted and structural duplication is applied.}
    \label{fig:Straight-shift_example}
\end{figure}

\begin{definition}[Straight-unshift $\mathrm{SU}_{n-k}$]\label{su} \normalfont
    Let $1 \le k \le n-2$ and $b=(x_1,\dots,x_{2n-1})\in\mathcal R^0$.
    Define $y=\mathrm{SU}_{n-k}(b)$ as:
    \[
    y_m=
    \begin{cases}
    x_{n-k}, & m=1,\\[2pt]
    x_{m-1}, & 1 < m \le n-k,\\[2pt]
    x_m,     & n-k<m\le n,\\[2pt]
    y_{m-n}, & n<m\le 2n-1.
    \end{cases}
    \]
\end{definition}

An example of the straight-unshift function at index \(n-2\) is shown below for the set \{$c_1, c_2, \dots, c_5$\}.

\begin{figure}[htbp]
    \centering
    \begin{tikzpicture}
        \node (b1) at (0,0) {\texttt{$c_1c_2c_3c_4c_5c_1c_2c_3c_4$}};
        \draw[->, thick] (b1.east) -- ++(1,0);
        \node (b2) at (4.2,0) {\texttt{$c_3c_1c_2c_4c_5c_3c_1c_2c_4$}};
    \end{tikzpicture}
    \caption{Application of straight-unshift \(\mathrm{SU}_{n-2}\) on a bead over \(\{c_1, c_2, \dots, c_5\}\), where the $n-2$th letter is re-inserted and structural duplication is applied.}
    \label{fig:Straight-unshift_example}
\end{figure}

\begin{lemma}[$\mathrm{SS}_{n-k}$ And $\mathrm{SU}_{n-k}$ Are Inverses]\label{ssinv}\normalfont
Let $1\le k\le n-2$. Then,
\[
\mathrm{SS}_{n-k}=(\mathrm{SU}_{n-k})^{-1}
\]
and
\[
\mathrm{SU}_{n-k}=(\mathrm{SS}_{n-k})^{-1}
\]
hold true.
\end{lemma}

\begin{proof}
Write $b=(x_1,\dots,x_n,\dots,x_{n-1})$.
Applying $\mathrm{SS}_{n-k}$, we obtain
\[
\mathrm{SS}_{n-k}(b) = (x_2,x_3,\dots,x_1,x_{n-k},\dots,x_n,\dots,x_{n-1}).
\]
Applying $\mathrm{SU}_{n-k}$ onto this bead, we obtain
\[
\mathrm{SU}_{n-k}(\mathrm{SS}_{n-k}(b)) = (x_1,\dots,x_n,\dots,x_{n-1}),
\]
proving the inverse in one direction. Similarly, apply $\mathrm{SU}_{n-k}$ to $b$ first, obtaining
\[
\mathrm{SU}_{n-k}(b) = (x_{n-k},x_1,x_2,\dots,x_{n-k-1},x_{n-k+1},\dots,x_n,\dots,x_{n-1}).
\]
Finally, by applying $\mathrm{SS}_{n-k}$ onto this bead, we obtain
\[
\mathrm{SS}_{n-k}(\mathrm{SU}_{n-k}(b)) = (x_1,\dots,x_n,\dots,x_{n-1}),
\]
proving the inverse in both directions. Hence, $\mathrm{SS}_{n-k}$ and $\mathrm{SU}_{n-k}$ are mutual inverses.
\end{proof}

The following notation will be introduced for ease throughout the paper.

\begin{notation}[Function Exponentiation]
For any function \(f\) on a bead \(b\) at index \(i\),
    \[
f_{i}^j(b)
\;:=\;
\underbrace{
  f_{i}\Bigl(
    f_{i}\bigl(\cdots(f_{i}(b))\bigr)
  \Bigr)
}_{j \text{ times}}.
\]
Note that this is the notation used in Axler \citep{bib3}; we adopt it to avoid clunky superscripts such as $\mathrm{SS}^{(j)}_{n-k}$.
\end{notation}

The following theorem formally relates our previously defined straight-shift function to the process of transitioning between different beads with a fixed intersection length corresponding to a straight-shift at a given index. In addition, it provides a preliminary basis for calculating the final length of the generated sequence via the number of total intersections.

\begin{theorem}[Intersection Properties Theorem]\label{ipt} \normalfont

The following properties hold for all intersections: 
    \begin{itemize}
        \item[(i)] The maximum intersection length between two beads is \(\ell=n-2\).
        \item[(ii)] We require \((n-1)!-1\) intersections in a superpermutation for \(n\).
    \end{itemize}
\end{theorem}
\begin{proof} Consider \(b_1 = (x_1,\dots,x_n,x_1,\dots,x_{n-1})\).

\medskip
\noindent\textbf{(i)}\textit{The maximum intersection length between two beads is \(\ell=n-2\).} 

\noindent We will first prove that any intersection length such that \(n \ge \ell > n-2\) is impossible by contradiction, and then directly prove that \(\ell=n-2\) is possible. 

Using the last \(n\) letters of $b_1$ as the overlap with $b_2$ can be ruled out, since this would result in a repeated permutation, contradicting Definition \ref{b}.

Now, assume an overlap of length \(n-1\). Then the suffix \(x_2\dots x_n\) of \(b_1\) would equal the prefix \(x_1\dots x_{n-1}\) of \(b_2\), forcing \(x_1=x_2\). This contradicts Definition \ref{b}, so an overlap of length \(n-1\) is impossible.

However, the last \(n-2\) letters of $b_1$ can be used as the start of the first permutation of $b_2$. To do so, we duplicate $b_1$ and straight-shift at the \(n-1\)th letter to create $b_2$. Then, $b_2$ starts with the unique permutation 

\[(x_2,x_3,\dots, x_{n-1},x_1,x_n).\]

Thus, the maximum possible intersection length between $b_1$ and $b_2$ is \(\ell=n-2\), since the letters \(x_2\) through \(x_{n-1}\) are shared.

\medskip
\noindent\textbf{(ii)}\textit{We require \((n-1)!-1\) intersections in a superpermutation for \(n\).} 

\noindent Thinking of beads as nodes and intersections as the segments that connect them, each new node needs one new segment to connect to the existing structure. However, the first node does not need to be connected as it is the start of the structure. Thus, there are \((n-1)!-1\) intersections: one less than the number of beads

\end{proof}
From the result of part (i) of the preceding theorem, we create the first rule for this framework.
\begin{prule}[Bead Transition Process]\label{btp} \normalfont 
When transitioning from a bead \(b_i\) to \(b_{i+1}\), set \(b_{i+1} = \mathrm{SS}_{n-1}(b_i)\)
(i.e.\ \(k=1\) in Definition \ref{ss}).
\end{prule}

Thus, this rule furnishes transitions from one bead to another while guaranteeing the maximum intersection length, a crucial feature that will be utilized in the construction of higher-level structures.

\subsection{1-rings}\label{1rings}

Now, we introduce the 1-ring, a structure that contains beads as substructures, building our way up to the entire superpermutation in the process.

\begin{definition}\label{1ringdef}[1-ring] \normalfont 
A 1-ring is the ordered cycle
\[
R^1=(b_1,b_2,\dots ,b_{n-1}),
\]
where $b_1$ is any bead and every later entry is the same bead straight-shifted by $\mathrm{SS}_{n-1}$:
\[
b_j = \mathrm{SS}^{j-1}_{n-1}(b_1) \ \mathrm{for} \ 1 < j \le n-1.
\]
Here, indices are taken modulo \(n-1\); after $b_{n-1}$, the cycle returns to $b_1$. We denote the set of all 1-rings as \(\mathcal{R}^1\).

\end{definition}

The following theorem will allow us to relate our previously defined straight-shift function to an incrementally larger structure.

\begin{theorem}[1-ring Properties]\label{thm:1ringprop} \normalfont

The following properties hold for any 1-ring \(R^1_i\):

    \begin{itemize}
        \item[(i)] In transitioning from a 1-ring $R_i^1$ to $R_{i+1}^1$, the greatest possible intersection length is \(\ell=n-3\).
        \item[(ii)] Each superpermutation for \(n\ge 4\) requires \((n-2)!\) 1-rings. 
    \end{itemize}
    
\end{theorem}

\begin{proof} We prove each component separately.

\medskip
\noindent\textbf{(i)}\textit{In transitioning from a 1-ring $R_i^1$ to $R_{i+1}^1$, the greatest possible intersection length is \(\ell=n-3\).} 

\noindent Consider the bead
\[b_1 = (x_1,\dots ,x_n,x_1,\dots,x_{n-1}).\]
Following \hyperref[btp]{Rule 1} \(n-2\) times, the bead 
\[b_{n-1} = (x_{n-1},x_1,\dots ,x_{n-2},x_n,\dots ,x_{n-2})\]
is reached, and this marks the end of $R^1_1$ \hyperref[1ringdef]{(Definition 7)}. 

If we apply \hyperref[btp]{Rule 1} once more (duplicate $b_{n-1}$ and straight‑shift at index $n-1$), we obtain
\[
b_n=(x_1,\dots,x_n,x_1,\dots,x_{n-1}),
\]
which is identical to the original bead $b_1$ in $R_1^1$.
Hence any overlap of length $n-2$ between two 1‑rings would force their leading beads to coincide, contradicting \hyperref[1ringdef]{Definition 7}. Thus, an intersection length of \(\ell=n-2\) between distinct 1-rings is impossible.

Now, we will examine if an intersection of length \(\ell=n-3\) is possible. In doing so, we will attempt to share the letters 
\[(x_2,x_3,\dots ,x_{n-2}).\]
Since we have cycled through all \(n-2\) distinct straight‑shifts at index \(n-1\) (see \hyperref[btp]{Rule 1}), any further application would return to a previously seen bead. Thus, we instead straight-shift $b_1$ in $R_1^1$ at the ($n-2$)-th index to obtain $b_1$ in $R_2^1$. Straight-shifting as prescribed yields 
\[(x_2,x_3,\dots, x_{n-2},x_1,x_{n-1},x_n,\dots,x_{n-1})\]
as $b_1$ of $R_2^1$, which has an intersection of length \(\ell=n-3\) with $R_1^1$.

\medskip
\noindent\textbf{(ii)}\textit{Each superpermutation for \(n\ge 4\) requires \((n-2)!\) 1-rings.} 

\noindent Each superpermutation for \(n\) requires \((n-1)!\) beads \hyperref[bp]{(Theorem 1)}, and each 1-ring contains \(n-1\) beads \hyperref[1ringdef]{(Definition 7)}. Thus, we require 
\[\frac{(n-1)!}{n-1} = (n-2)!\]
1-rings in the superpermutation.
\end{proof} 

\begin{remark} \normalfont
For \(n = 3\), the concept of a 1-ring degenerates to a single bead, hence we restrict \(n \ge 4\).
\end{remark}

The next rule follows from the proof of \hyperref[thm:1ringprop]{Theorem 3(i)}.

\begin{prule}[1-Ring Transition Process]\label{rule:1ringtrans}\normalfont
Let $R_i^1 = (b_1,\dots,b_{n-1})$ be a 1-ring, where each $b_j$ is a bead.  
To form the next 1-ring $R_{i+1}^1$, we:

\begin{enumerate}
    \item Apply $\mathrm{SS}_{n-2}$ to the leading bead $b_1$ of $R_i^1$, yielding a new bead $\tilde{b}_1 := \mathrm{SS}_{n-2}(b_1)$.
    \item Populate the remaining beads of $R_{i+1}^1$ by repeatedly applying $\mathrm{SS}_{n-1}$ to $\tilde{b}_1$:
    \[
      R_{i+1}^1 \;=\;
      \Bigl(\,\tilde{b}_1,\;
      \mathrm{SS}_{n-1}(\tilde{b}_1),\;
      \mathrm{SS}_{n-1}^2(\tilde{b}_1),\dots,
      \mathrm{SS}_{n-1}^{\,n-2}(\tilde{b}_1)\Bigr).
    \]
\end{enumerate}

(This yields \(n-1\) beads altogether.) Thus, $R_{i+1}^1$ is a 1-ring whose leading bead is given by $\mathrm{SS}_{n-2}(b_1)$, and whose subsequent beads follow our usual 1-ring definition
$\bigl(b_j = \mathrm{SS}_{n-1}^{\,j-1}(b_1)\bigr),$
but with $b_1$ replaced by $\tilde{b}_1$.
\end{prule}

Note that this rule is logically coherent with both \hyperref[thm:1ringprop]{Theorem 3(i)} and \hyperref[ipt]{Theorem 2(i)}, as it ensures an intersection length of \(\ell=n-3\) between any two 1-rings while ensuring that all beads within either 1-ring maintain an intersection length of \(\ell=n-2\).

\subsection{k-rings}\label{kopc}

Now, using a 1-ring as the base case, we inductively define higher orders of rings to construct the entire superpermutation structure.

\begin{definition}[k-ring]\label{kringdef} \normalfont

A k-ring is the ordered cycle
\[
R^k=(R^{k-1}_1,R^{k-1}_2,\dots ,R^{k-1}_{n-k}),
\]
where $R^{k-1}_1$ is any (k-1)-ring. We denote the set of all k-rings as $\mathcal{R}^k$.

Let \(n \ge 3\). We define a k-ring for integers \(1 \le k \le n-3\) as follows:

\begin{itemize}
    \item[(a)] \textit{Base Case}
    
    \noindent
    The base case for this inductive construction will be the 1-ring, which is defined as an ordered cycle of beads
    \[
    R^1=(b_1,b_2,\dots ,b_{n-1}),
    \]
    where
    \[
    b_j = \mathrm{SS_{n-1}^{j-1}}(b_1) \ \text{for } 1 < j \le n-1.
    \]
    (As stated previously, this transformation wraps mod $n-1$). The transition norm for 1-rings will also be utilized as the base case for the inductive construction as well.
    \item[(b)] \textit{Transition between k-rings for \(k \ge 2\)}

    \noindent
    Now, we define how to "transition" from one (\(k-1\))-ring to another, which we will utilize to construct our k-ring. Denote the (\(k-1\))-ring transition operator as \(\mathrm{Trans}(\cdot)\).

    Given a \((k-1)\)-ring \(R^{k-1}_i\), we descend to the leading bead \(b\) of the structure (obtained by recursively descending through the first \(k-1\) subring) and straight shift it at index \(n-(k-1)-1 = n-k\) to obtain a new anchor bead \(b'\). Build a new \((k-1)\)-ring from scratch through the same process that yielded \(R^{k-1}_i\) and call it \(R^{k-1}_{i+1}\). Specifically, \(R^{k-1}_{i+1} := \mathrm{Trans}(R^{k-1}_i)\). Similar to 1-rings, the $\mathrm{Trans}(\cdot)$ operator wraps mod $n-k$, where we have a k-ring constituted of a cycle of $k-1$-rings.

    \item[(c)] \textit{k-ring from scratch}

    \noindent Finally, for \(k \ge 2\), to build a k-ring from scratch we first build the initial subring \(R^{k-1}_1\) from scratch. Then, for \(2 \le j \le (n-k)\), define
    \[
    R^{k-1}_j = \mathrm{Trans}(R^{k-1}_{j-1}),
    \]
    where each step uses the transition to go from the previous subring to the next. Finally, set
    \[
    R^k:=(R^{k-1}_1,R^{k-1}_2,\dots ,R^{k-1}_{n-k}),
    \]
    completing the construction.
\end{itemize}

\end{definition}

To clear up any ambiguity about this procedure, we illustrate the process for creating a 2-ring around the initial bead \(b = (1,2,3,4,1,2,3)\) (here \(n=4\)).

\begin{example}[2-ring creation]

\noindent
To begin, we fall to our base case to create the first 1-ring in the cycle. To do so, we must obtain the ordered cycle
\[
R^1_1 = (b, \mathrm{SS}_{n-1}(b), \mathrm{SS}^2_{n-1}(b)).
\]
Performing the computation according to the definition of \hyperref[ss]{straight-shift} yields
\[
R^1_1 = ((1,2,3,4,1,2,3),(2,3,1,4,2,3,1),(3,1,2,4,3,1,2)).
\]
Next, we must transition from \(R^1_1\) to \(R^1_2\). Following the defined transition rule \hyperref[rule:1ringtrans]{(Rule 2)} yields a new leading bead 
\[
b' = (2,1,3,4,2,1,3)
\]
for \(R^1_2\). Following the same procedure as \(R^1_1\), we obtain
\[
R^1_2 = ((2,1,3,4,2,1,3),(1,3,2,4,1,3,2),(3,2,1,4,3,2,1)).
\]
Finally, we put these 1-rings together to obtain the final 2-ring:

\begin{align*}
R^2 =\; &\Bigl(\Bigl(
(1,2,3,4,1,2,3),\,(2,3,1,4,2,3,1),\,(3,1,2,4,3,1,2)
\Bigr),\\[5pt]
&\quad\Bigl(
(2,1,3,4,2,1,3),\,(1,3,2,4,1,3,2),\,(3,2,1,4,3,2,1)
\Bigr)\Bigr).
\end{align*}

\end{example}

Note that any creation of a k-ring (where $k \ge 2$) would function similarly, in that it is constructed from descending to a constituent 1-ring and constructing around it as prescribed from \hyperref[kringdef]{Definition 8}.

Now, we introduce a new notation that will be required to prove properties about k-rings.

\begin{notation}[Abuse of Notation and Overlap Merging]
Throughout this paper, we may sometimes denote a $k$-ring $R^k$ as an ordered cycle of its constituent $(k-1)$-rings (or beads), or alternatively as a single flattened sequence in which any overlapping segments are merged. For example, in the case of a 1-ring $R^1$ over the set $\{1,2,3,4\}$, one might write

\[
R^1 =
\bigl(
(1,2,3,4,1,2,3),
(2,3,1,4,2,3,1),
(3,1,2,4,3,1,2)
\bigr)
\]

to emphasize its structure as a cycle of three beads. However, for brevity, the same ring can also be represented by the `flattened' sequence

\[
(1,2,3,4,1,2,3,\,1,4,2,3,1,2,4,3,1,2),
\]

where internal commas and parentheses are omitted and the overlapping segments shared between consecutive rings (or beads) are merged.

From this point on, the `abuse' of merging overlaps in a flattened representation will be utilized whenever it is more convenient, with all details provided so there is no ambiguity.
\end{notation}

For convenience, we add another name to identify beads.

\begin{notation}[0-ring] \normalfont

A bead, \(b\), can be equivalently denoted as a 0-ring, \(R^0\).

\end{notation}

We also formalize the notion of a `leading bead' that was utilized in defining a k-ring.

\begin{definition}[Leading bead] \normalfont

The leading bead of any $R^k$ is the first $R^0$ in the sequence.

\end{definition}

Now, we can prove desired properties about k-rings.

\begin{theorem}[k-ring Properties]\label{krp} \normalfont
    
    The following properties hold for all k-rings:
    
    \begin{itemize}
        \item[(i)] We require \((n-(k+1))!\) k-rings in a superpermutation.
        \item[(ii)] Any k-ring begins and ends with the same \(n-(k+1)\) letters.
        \item[(iii)] The maximum intersection length between two k-rings is \(\ell = (n-(k+2))\).
    \end{itemize}
    
\end{theorem}

\begin{proof} Each of these components will be proven separately.

\medskip
\noindent\textbf{(i)}\textit{We require \((n-(k+1))!\) k-rings in a superpermutation} 

\noindent We will prove this via induction for a given \(n\). As the base case, we know that we require \((n-1)!\) 0-rings \hyperref[bp]{(Theorem 1)}. 

Thus, our inductive hypothesis is as follows: we require \((n-(k+1))!\) k-rings. Now, for the inductive step, we must prove that we require \((n-(k+2))!\) \((k+1)\)-rings. Using our inductive hypothesis, we can divide the required number of k-rings by the amount of k-rings in each \((k+1)\)-ring \hyperref[kringdef]{(Definition 8)}. This yields
\[
\frac{(n-(k+1))!}{n-(k+1)} = (n-(k+2))!,
\]
thus proving through induction for all \(0 \le k \le n-3\).

\medskip
\noindent\textbf{(ii)}\textit{Any k-ring begins and ends with the same \(n-(k+1)\) letters.} 

\noindent
We will prove this via induction for a given \(n\). As the base case, we know that any 0-ring begins and ends with the same \(n-1\) letters \hyperref[bp]{(Theorem 1)}. 

Our inductive hypothesis is as follows: any k-ring begins and ends with the same \(n-(k+1)\) letters. Now, for the inductive step, we must prove that any $k+1$-ring begins and ends with the same \(n-(k+2)\) letters. 

Suppose the leading bead of $R^{k+1}$ is 

\[
b = (x_1,x_2,\dots ,x_n, \dots , x_{n-1}).
\] 

We will now determine the leading bead of the last k-ring in this $k+1$-ring. We know there are \(n-(k+1)\) k-rings in this $k+1$-ring \hyperref[kringdef]{(Definition 8)}, so we must perform \(n-(k+1)-1 = n-k-2\) Straight-shifts at index \(n-k-1\) to obtain this bead (recall that the leading bead of $R^{k+1}$ is also the leading bead of $R^k_1$).

Applying this transformation yields the leading bead of $R^k_{n-k-1}$:
\[
\tilde{b} = (x_{n-k-1},x_1,x_2,\dots ,x_{n-k-2},x_{n-k},x_{n-k+1},\dots ,x_n,\dots ,x_{n-1})
\]
From the inductive hypothesis, we know that $R^k_{n-k-1}$ will begin and end with the same \(n-k-1\) letters (here we utilize the aforementioned flattened notation):
\[
R^k_{n-k-1} = (x_{n-k-1},x_1,x_2,\dots , x_{n-k-2},\dots \dots ,x_{n-k-1},x_1,x_2,\dots , x_{n-k-2}).
\]
As such, we see that $R^{k+1}$ begins and ends with the same 
\[
n-k-2 = n-(k+2)
\]
letters as desired, thus proving through induction for all \(0 \le k \le n-3\).

\medskip
\noindent\textbf{(iii)}\textit{The maximum intersection length between two k-rings is \(\ell = (n-(k+2))\).} 

\noindent We will proceed via strong induction. For our base case, we know that 0-rings intersect with length \(\ell = n-2\) \hyperref[ipt]{(Theorem 2)}. 

For the inductive hypothesis, assume that all j-rings, such that \(0 \le j \le k\) intersect with length \(\ell = n-(k+2)\). Furthermore, we utilize the definition of transitioning between k-rings \hyperref[kringdef]{(Definition 8(b))}. 

From (ii), we write

\[
R_1^{k+1} = (x_1,x_2,\dots ,x_{n-k-2}, \dots \dots ,x_1,x_2,\dots , x_{n-k-2}).
\]

Observe that we have already exhausted all possible permutations starting with \((x_1,x_2,\dots ,x_{n-k-2})\) due to the Straight-shift process at every \(m\)th letter, such that \(n-k-3 < m \le n-1\), performed by transitioning through all j-rings, such that \(0 \le j \le k\). 

Thus, \(R_2^{k+1}\) can, at most, share the letters \((x_2,x_3,\dots , x_{n-k-2})\) with $R^{k+1}_1$. As before, we duplicate the leading bead of $R^{k+1}_1$ and straight-shift at the \(n-(k+1)-1\)th letter. In doing so, \(R_2^{k+1}\) starts with the unique permutation

\[
(x_2,x_3,\dots, x_{n-k-2},x_1,x_{n-k-1},\dots,x_n),
\]

establishing intersections of length \(\ell = n-(k+3)\) for $k+1$-rings, as desired. Thus, the proof is completed via strong induction for all \(0 \le k \le n-3\).
\end{proof} 

Now, we formalize the rule that we utilized for the definition and proofs in the same format as previous orders of rings for future reference.

\begin{prule}[k-ring Transition Process]\label{kringtrans} \normalfont

Let \(R^k_i\) = \((R^{k-1}_1,\dots ,R^{k-1}_{n-k})\) be a k-ring, where each $R^{k-1}_j$ is a \((k-1)\)-ring. To form the next k-ring $R^k_{i+1}$, we:

\begin{itemize}
    \item[1.] Apply \(\mathrm{SS}_{n-k-1}\) to the leading bead \(b\) of \(R^k_i\), yielding a new bead \(\tilde{b} := \mathrm{SS}_{n-k-1}(b)\).
    \item[2.] Populate the remaining \((k-1)\)-rings in the ordered cycle by using \(\tilde{b}\) as the anchor bead to create \(R^k_{i+1}\) from scratch
\end{itemize}
   
\end{prule}

\subsection{The Structure}

The following tree diagram illustrates the relationship between the various ring orders for a given \(n\) and its associated superpermutation \(W_n\). The depicted relationship continues indefinitely while \(n-k \ge 0\).

\begin{figure}[htbp]
  \centering
  \includegraphics[width=1\linewidth]{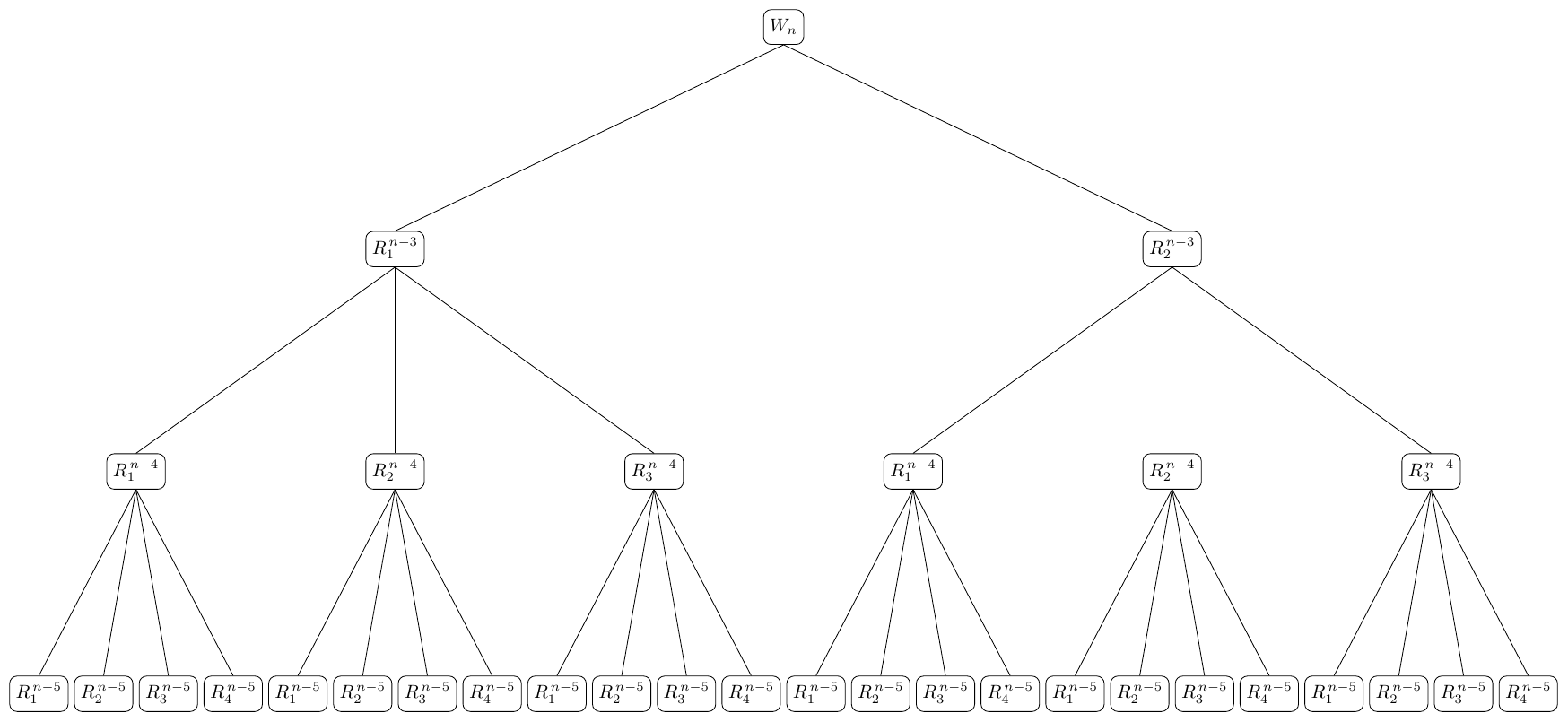}
  \caption{Hierarchical relationship between \(k\)-rings.}
  \label{fig:k_structure}
\end{figure}

\noindent
The following definitions will be utilized to prove the next theorem.
\begin{definition}[Mirroring Operation] \normalfont
Let $s$ be a finite string indexed from $1$ to $\ell$. We define its \emph{mirroring operation} as follows:

\begin{itemize}
    \item \textbf{Case 1: $\ell = 2n + 1$ is odd.}  
    Let $i = n + 1$ be the central index. For any index $m \in \{1,2,\ldots,\ell\}$, the mirroring operation maps the character at index $m$ to the position
    \[
    i + (i - m).
    \]
    In other words, we reflect each character across the central index $i$.

    \item \textbf{Case 2: $\ell = 2n$ is even.}  
    Let the two central indices be $i = n$ and $j = n + 1$. Then:
    \begin{itemize}
      \item If $m \le i$, the character at index $m$ is mapped to position 
      \[
      i + (j - m).
      \]
      \item If $m \ge j$, the character at index $m$ is mapped to position 
      \[
      j + (i - m).
      \]
    \end{itemize}
\end{itemize}
\end{definition}

To quickly reference this operation, we introduce the following notation.

\begin{notation}[Mirror Notation] \normalfont

The mirror of some string \(s\) is denoted as \(s^*\).

\end{notation}

Using this new notion of a mirror, we introduce the following functions to aid in more efficient construction of the superpermutation, as will be proven later.

\begin{definition}[Mirror‑shift $\mathrm{MS}_{n-k}$]\label{ms} \normalfont
    Let $1\le k\le n-2$ and $R^0=(x_1,\dots,x_{2n-1})\in\mathcal R^0$.
    Define $y=\mathrm{MS}_{n-k}(R^0)$ as:
    \[
    y_m=
    \begin{cases}
    x_{m+k+1}, & 1\le m \le n-k-2,\\[2pt]
    x_{n-m},     & n-k-2 < m < n,\\[2pt]
    x_m, & m=n,\\[2pt]
    y_{m-n}, & n<m\le 2n-1.
    \end{cases}
    \]
\end{definition}
Below, the mirror-shift function at the \(n-1\)th letter for the set \{$c_1, c_2, \dots, c_5$\} is depicted.
\begin{figure}[htbp]
    \centering
\begin{tikzpicture}
  \node (p1) at (0,0) {\texttt{$c_1c_2c_3c_4c_5c_1c_2c_3c_4$}};
  \draw[->, thick] (p1.east) -- ++(1,0);
  \node (p2) at (4.2,0) {\texttt{$c_3c_4c_2c_1c_5c_3c_4c_2c_1$}};
\end{tikzpicture}
\caption{Application of mirror-shift \(\mathrm{MS}_{n-1}\) on a bead over \(\{c_1, c_2, \dots, c_5\}\), where letters are inserted and mirrored about the $n-1$-th index; structural duplication is applied.}
    \label{fig:mirror-shift_example}
\end{figure}

\begin{definition}[Mirror‑unshift $\mathrm{MU}_{n-k}$]\label{mu} \normalfont
    Let $1\le k\le n-2$ and $R^0=(x_1,\dots,x_{2n-1})\in\mathcal R^0$.
    Define $y=\mathrm{MU}_{n-k}(R^0)$ as:
    \[
    y_m=
    \begin{cases}
    x_{n-m}, & 1\le m< k+2,\\[2pt]
    x_{m-(k+1)},     & k+2 \le m < n,\\[2pt]
    x_m, & m=n,\\[2pt]
    y_{m-n}, & n<m\le 2n-1.
    \end{cases}
    \]
\end{definition}
Below, the mirror-unshift function at the \(n-1\)th letter for the set \{$c_1, c_2, \dots, c_5$\} is depicted.
\begin{figure}[htbp]
    \centering
    \begin{tikzpicture}
        \node (p1) at (0,0) {\texttt{$c_1c_2c_3c_4c_5c_1c_2c_3c_4$}};
        \draw[->, thick] (p1.east) -- ++(1,0);
        \node (p2) at (4.2,0) {\texttt{$c_4c_3c_1c_2c_5c_4c_3c_1c_2$}};
    \end{tikzpicture}
    \caption{Application of mirror-unshift \(\mathrm{MU}_{n-1}\) on a bead over \(\{c_1, c_2, \dots, c_5\}\), where letters are re-inserted and mirrored about the $n-1$-th index; structural duplication is applied.}
    \label{fig:mirror-unshift_example}
\end{figure}

Similar to the formalized notion of a leading bead, we now introduce the concept of a trailing bead.

\begin{definition}[Trailing bead] \normalfont

The trailing $R^0$ of some $R^k$ is the last bead in that k-ring.

\end{definition}

To prove the improved efficiency of these functions, we first must prove the following lemma and theorem.

\begin{lemma}[Straight-shift Equivalence Lemma]\label{lemma:sseq} \normalfont

    Straight-shifting \(j-2\) times at the \((j-1)\)th index on some bead is equivalent to making the \(j-1\)-th letter the leading letter of that bead, moving each letter of index \(m\) (such that \(1 \le m \le j-2\)), to index \(m+1\). Every letter at index $m\;(n<m\le2n-1)$ is replaced by the letter at index $m-n$.
    
\end{lemma}
\begin{proof} \normalfont
    Consider an arbitrary bead
    \[
    R^0 = (x_1,x_2,\dots ,x_n,\dots,x_{n-1}).
    \]
    Select an arbitrary index \(j-1\) as the straight-shift location. Applying $j-2$ straight-shifts yields the bead
    \[
    \mathrm{SS}_{j-1}^{j-2}(R^0)=\bar{R^0} = (x_{j-1},x_1,x_2,\dots,x_{j-2},x_j,x_{j+1},\dots,x_n,\dots\dots,x_{n-1}).
    \]
    Hence, the letter formerly at index $j-1$ becomes the new leading letter and every letter at index $m\;(1\le m\le j-2)$ shifts to $m+1$, proving the claimed equivalence.
\end{proof}

Now, we utilize this result to prove the theorem.

\begin{theorem}[Trailing bead Theorem]\label{thm:trailingbeadtheorem}\normalfont 

Assuming a leading bead \(R^0=(x_1,\dots ,x_n,\dots ,x_{n-1})\) within some k-ring, the trailing bead will be of the form 
\[
(x_{n-1},x_{n-2},\dots, x_{n-k+1},x_{n-k},x_1,x_2,\dots,x_{n-k-1},x_n,\dots\dots,x_{n-k-1}).
\]
\end{theorem}

\begin{proof} \normalfont 
Consider some arbitrary k-ring $R^k$ whose leading bead is equal to $R^0$ from the previous proof. To obtain the trailing bead, we must traverse through all $R^j$, such that \(0 \le j \le k\) by performing the straight-shift operation at all respective indices. For each order, \(j\), we will straight-shift at the \(n-j-1\)th letter \hyperref[kringtrans]{(Rule 3)}, and we have \(n-j-2\) available successive unique straight-shifts at this point.

To apply the straight-shifts for the $k-1$-rings, we use \hyperref[lemma:sseq]{Lemma 2} to move the \(n-k\)th letter to the beginning of the string. Thus, we obtain the leading bead \(\tilde{b}\) of \(R^{k-1}_{n-k}\):

\[
\tilde{b} = (x_{n-k}, x_1,x_2,\dots,x_{n-k-1},x_{n-k+1},\dots,x_n,\dots\dots,x_{n-1}).
\]

Repeating this process for all subsequent j-rings (\(j < k\)), we obtain the trailing bead of $R^k$:

\[
\bar{b} = (x_{n-1},x_{n-2},\dots, x_{n-k+1},x_{n-k},x_1,x_2,\dots,x_{n-k-1},x_n,\dots\dots,x_{n-k-1})
\]

Thus, the proof is complete.
\end{proof}

With this proof, we can now prove important properties about $\mathrm{MU}_{n-k}$ and $\mathrm{MS}_{n-k}$. However, we first set the groundwork for function composition in this paper.

\begin{notation}[Function Composition] \normalfont
    The following notation will be utilized.
    \begin{itemize}
        \item[(i)] The composition of any two functions is denoted by the symbol \(\circ\).
        \item[(ii)] Repeated composition of a function is denoted by \(\bigcircop\). `Regular' ordering (higher value on top and lower on bottom) indicates an orthodox output, akin to a summation. For example,
        \[
        \bigcircop_{k=1}^jf_k=f_1 \circ f_2 \circ f_3 \circ \cdots \circ f_j.
        \]
        Inverted ordering (lower value on top and higher on bottom) indicates that the output will run in reverse. For example,
        \[
        \bigcircop_{k=j}^1f_k=f_j \circ f_{j-1} \circ f_{j-2} \circ \cdots \circ f_1.
        \]
        Orthodox function composition norms apply (non-commutativity, right-to-left composition).
    \end{itemize}
    
\end{notation}

Now, we can establish the desired identities.

\begin{theorem}[Function Identities]\label{fi}\normalfont

The following identities hold for all \(n, k\).

\begin{itemize}
        \item[(i)] Mirror-unshift equals straight-unshift coupled with straight-shifts:
        \[\mathrm{MU}_{n-k} = \bigl( \displaystyle\bigcircop_{j=k+1}^{0}\;\displaystyle\bigcircop_{l=1}^{n-k-2+j}\mathrm{SS}_{n-k+j-2}\bigr) \circ\mathrm{SU}_{n-k-1}.\] 
        \item[(ii)] Mirror-shift equals straight-unshifts coupled with straight-shift:
        \[\mathrm{MS}_{n-k} =  \mathrm{SS}_{n-k-1} \circ\displaystyle\bigcircop_{j=k+1}^{0}\displaystyle\;\bigcircop_{l=1}^{n-j-2}\mathrm{SU}_{n-j-1}.\]
\end{itemize}

\end{theorem}

\begin{proof}\normalfont Each of these will be proven separately.

\medskip
\noindent\textbf{(i)}\textit{Mirror-unshift equals straight-unshift coupled with straight-shifts.}

\noindent
First, we rewrite the double composition as the following equivalent expression:
\[
\displaystyle\bigcircop_{l=1}^{n-2}\mathrm{SS}_{n-1} \circ \displaystyle\bigcircop_{l=1}^{n-3}\mathrm{SS}_{n-2} \circ \cdots \circ \displaystyle\bigcircop_{l=1}^{n-k-3}\mathrm{SS}_{n-k-2}.
\]

Since $l$ solely serves as an iterator, we can rewrite the expression once more:
\[
\mathrm{SS}_{n-1}^{n-2} \;\circ\; \mathrm{SS}_{n-2}^{n-3} \;\circ\; \cdots \;\circ\; \mathrm{SS}_{n-k-2}^{n-k-3}.
\]
Observe that this is simply an expression to obtain the trailing bead of some k-ring \hyperref[thm:trailingbeadtheorem]{(Proof of Theorem 5)}. 

Consider the two arbitrary k-rings $R^k_i$ and $R^k_{i+1}$, with the leading bead of $R^k_i$ being
\[
R^0 = (x_1,x_2,\dots,x_n,\dots,x_{n-1}).
\]
Thus, the leading bead of $R^k_{i+1}$ is 
\[
\mathrm{SS}_{n-k-1}(R^0)=\tilde{R^0} = (x_2,x_3,\dots,x_1,x_{n-k},x_{n-k+1},\dots,x_n,\dots,x_{n-1}).
\]
Starting with \(\tilde{R^0}\), we show that the two functions are equivalent.

Mirror-unshifting this bead yields
\[
\mathrm{MU}_{n-k}(\tilde{R^0})=\bar{R^0} = (x_{n-1},x_{n-2},\dots,x_{n-k+1},x_{n-k},x_1,x_2,\dots,x_{n-k-2},x_{n-k-1},x_n,\dots,x_{n-1}).
\]
Now, we apply the other transformation. Straight-unshifting yields
\[
\mathrm{SU}_{n-k-1}(\tilde{R^0}) = R^0 = (x_1,x_2,\dots,x_n,\dots,x_{n-1}),
\]
the leading bead of $R^k_i$. Now, by \hyperref[thm:trailingbeadtheorem]{Theorem 5}, the trailing bead of this k-ring is 
\[
R^{0'} = (x_{n-1},x_{n-2},\dots,x_{n-k+1},x_{n-k},x_1,x_2,\dots,x_{n-k-2},x_{n-k-1},x_n,\dots,x_{n-1}).
\]
Since $R^{0'} = \tilde{R^0}$, equivalence is proven.

\medskip
\noindent\textbf{(ii)}\textit{Mirror-shift equals straight-unshifts coupled with straight-shift.}

\noindent We rewrite the double composition as the following equivalent expression:
\[
\displaystyle\bigcircop_{l=1}^{n-k-3}\mathrm{SU}_{n-k-2} \circ \displaystyle\bigcircop_{l=1}^{n-k-2}\mathrm{SU}_{n-k-1} \circ \cdots \circ \displaystyle\bigcircop_{l=1}^{n-2}\mathrm{SU}_{n-1}.
\]
Since $l$ solely serves as an iterator, we can rewrite the expression once more:
\[
\mathrm{SU}_{n-k-2}^{n-k-3} \;\circ\; \mathrm{SU}_{n-k-1}^{n-k-2} \;\circ\; \cdots \;\circ\; \mathrm{SU}_{n-1}^{n-2}.
\]
Observe that this is simply the reversal of \hyperref[thm:trailingbeadtheorem]{(Proof of Theorem 5)}: an expression to obtain the leading bead of some k-ring from the trailing bead of that k-ring. 

Consider the two arbitrary k-rings from part (i). Specifically, assume the trailing bead of $R^k_i$ to be 
\[
\tilde{R^0} = (x_{n-1},x_{n-2},\dots,x_{n-k+1},x_{n-k},x_1,x_2,\dots,x_{n-k-2},x_{n-k-1},x_n,\dots\dots,x_{n-1}).
\]
Applying this compound function to $\tilde{R^0}$ yields the leading bead of $R^k_i$: 
\[
\mathrm{SU}_{n-k-2}^{n-k-3} \;\circ\; \mathrm{SU}_{n-k-1}^{n-k-2} \;\circ\; \cdots \;\circ\; \mathrm{SU}_{n-1}^{n-2}(\tilde{R^0})=R^0 = (x_1,x_2,\dots,x_n,\dots,x_{n-1}).
\]
After straight-shifting $R^0$, we obtain the leading bead of $R^k_{i+1}$
\[
\mathrm{SS}_{n-k-1}(R^0)=\bar{R^0} = (x_2,x_3,\dots,x_1,x_{n-k},x_{n-k+1},\dots,x_n,\dots,x_{n-1}).
\]
Conversely, mirror-shifting the trailing bead of $R^k_i$ yields
\[
\mathrm{MS}_{n-k}(\tilde{R^0})=R^{0'} = (x_2,x_3,\dots,x_1,x_{n-k},x_{n-k+1},\dots,x_n,\dots,x_{n-1}).
\]
Since $R^{0'} = \tilde{R^0}$, equivalence is proven.
\end{proof}

Now, we can effectively transition between two k-rings without requiring the storage of a leading bead, since we can simply mirror-shift the trailing bead of any $R^k_i$ to obtain the leading bead of $R^k_{i+1}$, enabling the construction from there. 

Similarly, we can obtain the leading bead of any $R^k_i$ by mirror-unshifting the trailing bead of $R^k_{i+1}$. This allows us to conclude that $\mathrm{MS}_{n-k}$ and $\mathrm{MU}_{n-k}$ are mutual inverses, an idea we formalize in the following corollary.

\begin{corollary}[$\mathrm{MS}_{n-k}$ and $\mathrm{MU}_{n-k}$ are Inverses]
    Let \(1 \le k \le n-2\). Then,
    \[
    \mathrm{MS}_{n-k} = (\mathrm{MU}_{n-k})^{-1}
    \]
    and
    \[
    \mathrm{MU}_{n-k} = (\mathrm{MS}_{n-k})^{-1}
    \]
    hold true.
\end{corollary}

\begin{proof}
    The result directly follows from the proof of \hyperref[fi]{Theorem 6}, as we show that the compound functions are inverses of each other, and $\mathrm{SS}_{n-k}$ and $\mathrm{SU}_{n-k}$ are known inverses \hyperref[ssinv]{(Lemma 1)}.
\end{proof}
More importantly, these new functions facilitate a further reduction in space requirements, as we can simply store the current bead and mirror-shift it through all k-rings (\(1\le k \le n-3\)) rather than storing the leading bead of every k-ring.

The following theorem furnishes an equivalent tradeoff between practical runtime and memory usage as a natural consequence of the structure that we generate via this procedure.

\begin{theorem}[Palindromic Nature of the Resulting Superpermutation] \normalfont
$R_2^{n-3}$ is the mirror of $R_1^{n-3}$.
\end{theorem}

\begin{proof}
Consider the arbitrary $n-3$-ring $R_1^{n-3}$ with leading bead \(R^0 = (x_1,x_2,\dots,x_n,\dots,x_{n-1})\).
We will first determine the trailing bead $\tilde{R^0}$ of $R_1^{n-3}$. Applying \hyperref[thm:trailingbeadtheorem]{Theorem 5} yields: 
\[
\tilde{R^0} = (x_{n-1},x_{n-2},\dots,x_3,x_1,x_2,x_n,\dots,x_2).
\]
To determine the leading bead of $R_2^{n-3}$, we simply take the leading bead of $R_1^{n-3}$ and straight-shift at the 2nd character \hyperref[kringtrans]{(Rule 3)}. Doing so yields:
\[
\mathrm{SS}_2(R^0)=R^{0'} = (x_2,x_1,x_3,x_4,\dots,x_{n-1},x_n,\dots,x_{n-1}).
\]
As can be seen, these sequences are mirrors of each other. This will be used this as our base case.

Now, consider two arbitrary beads that are mirrors of each other. Without loss of generality, let
\[
\bar{R^0} = (x_1,x_2,\dots,x_n,\dots\dots,x_{n-1})
\]
and
\[
(\bar{R^0})^* = (x_{n-1},x_{n-2},\dots,x_n,\dots\dots,x_1).
\]
Consider the arbitrary index to mirror-unshift and mirror-shift to be the \(n-k\)th letter.

Mirror-unshifting $(\bar{R^0})^*$ yields 
\[
\mathrm{MU}_{n-k}((\bar{R^0})^*)=\hat{R^0} = (x_1,x_2,\dots,x_k,x_{k+1},x_{n-1},\dots,x_{k+3},x_{k+2},x_n,\dots,x_{k+2}),
\] 
and mirror-shifting $\bar{R^0}$ yields 
\[
\mathrm{MS}(\bar{R^0})=\breve{R^0} = (x_{k+2},x_{k+3},\dots,x_{n-1},x_{k+1},x_k,\dots,x_2,x_1,x_n,\dots,x_1).
\]
Thus, the mirror property is preserved after $\mathrm{MU}$ and $\mathrm{MS}$ for any arbitrary bead and its mirror, respectively, at any index.

In the superpermutation, since we start with two initial beads that are mirrors, any mirror-unshift or mirror-shift as we traverse through the first or latter half of the sequence, forward or backward, respectively, will preserve the mirror property. This ends with the leading bead $\breve{R^0}$ of $R^{n-3}_1$, being the mirror of the trailing bead $\hat{R^0}$ of $R^{n-3}_2$. 

Specifically,
\[
\breve{R^0} = (x_1,x_2,\dots,x_n,\dots,x_{n-1})
\]
and
\[
\hat{R^0} = (x_{n-1},x_{n-2},\dots,x_2,x_1,x_n,\dots,x_1),
\]
where we derived $\hat{R^0}$ using \hyperref[thm:trailingbeadtheorem]{Theorem 5} on the leading bead of $R^{n-3}_2$.

Thus, 
\[
R_2^{n-3} = (R_1^{n-3})^*,
\]
making this structure a palindrome and completing the proof.
\end{proof} 

\section{The Algorithm}
\subsection{Mathematical Logic}
To begin, we will outline the algorithm mathematically. Consider the \(+\) sign to denote concatenation rather than addition of numbers. The input is a set of \(n\) distinct letters: 
\[
s = \{x_1, x_2, x_3,\dots, x_n\}.
\]
From this input, we define our initial bead
\[
b_o = (x_1,x_2,\dots,x_n,\dots\dots,x_{n-1})
\]
by concatenating \(s\) with all but the last letter of \(s\).

We establish an integer \(\ell\) to denote the intersection length, which is initially set to the value \(n-2\). 

Next, we define a function \(\phi_{n-k}\), which takes the output of \(c_{n-k}\) and outputs the substring from \(\ell\) to the length of the string.

Each successive \(\phi\) function takes the output bead of the previous \(\phi\) function as its input bead, with the first call to the \(\phi\) function using \(b_o\) as its input.

Thus, the superpermutation \(W_n\) can be constructed:

\begin{align*}
W_n &= b_o + 
\sum_{x=0}^{1}\Bigg(
    \sum_{y=0}^{2}\Bigg(
        \sum_{z=0}^{3}\Bigg(
            \dots
            \sum_{\kappa=0}^{n-3}\Bigg(
                \left(
                    \sum_{\gamma=0}^{n-2} 
                    \phi_{n-1}|_{n-2} \times 1_{\gamma<n-2}
                \right) \\
                &\quad + \phi_{n-2}|_{n-3} \times 1_{\kappa<n-3}
            \Bigg) 
            \dots + \phi_4|_{3} \times 1_{z<3}
        \Bigg) 
        + \phi_3|_{2} \times 1_{y<2}
    \Bigg) \\
    &\quad + \phi_2|_{1} \times 1_{x<1}
\Bigg),
\end{align*}

where \(1_{x<f}\) serves as an indicator function that ensures that \(\phi\) is only performed when \(x<f\), and \(|_{\ell=k}\) serves as an indication of the intersection length to be considered by the \(\phi\) function.

Now, we calculate the total number of operations.

\begin{theorem}[Operation Count]\normalfont
    The total number of operations in our algorithm for a given \(n\) is \((n-1)!-1\).
\end{theorem}

\begin{proof}
    First, we define the function \(T(d)\), which returns the number of operations at each level \(d\), starting with \(d=2\) until \(d=n\) (based on the number of iterations of the loops).

    At each level, \(d\), we:
    \begin{itemize}
        \item Perform \(d\) recursive calls.
        \item Call the \(\phi\) function \(d-1\) times.
    \end{itemize}
    The edge case is where d=n, as the number of operations goes to 0 (\(T(n) = 0\)). Thus, we establish the recursive relation
    \[
    T(d) = d \times T(d+1) + (d-1) \ \text{for} \ 2\le d \le n-1.
    \]
    Now, for convenience, we define a new function
    \[
    S(d) = T(d) + 1.
    \]
    Thus, 
    \[
    T(d) = S(d) -1.
    \]
    Substituting this into the original equation and simplifying yields 
    \[
    S(d)=d \times S(d+1).
    \]
    The boundary condition for T gives that
    \[
    T(n) = 0 \implies S(n) = T(n) + 1 = 1.
    \]
    Now, we expand the recurrence relation on S:
    \[
    S(d) = d \times S(d+1) = d(d+1)S(d+2)= \prod_{j=d}^{n-1}j\times S(n).
    \]
    Since \(S(n)=1\), it immediately follows that
    \[
    S(d) = \prod_{j=d}^{n-1}j.
    \]
    Expressed in factorial notation, this product is 
    \[
    S(d) = \frac{(n-1)!}{(d-1)!}.
    \]
    Substituting $T$ back in for $S$ yields
    \[
    T(d) = \frac{(n-1)!}{(d-1)!}-1. 
    \]
    Finally, we are interested in the number of operations from \(d=2\). Substituting yields that there are
    \[
    T(2) = (n-1)! -1
    \]
    total operations in our algorithm.
\end{proof}

In the subsequent subsection, we extend this procedure to provide a Java coded implementationof the algorithm. 

\subsection{Java Implementation}
For the sake of simplicity, assume the \(n\) symbols are characters, entered into the program in array format. 

 We rewrite the iterative nested loops in a recursive fashion, having the function call itself at each iteration of a loop that runs a unique number of times based on our position in the recursive stack.
\begin{lstlisting}[basicstyle={\ttfamily}]
public void algo(int d) {
// When initially called, feed in d=2
      if (d>=n) {return;}   
      for (int i = 0; i < d; i++) {
         algo(d+1);
         intersect = d-1;
         if (i == d-1) {return;}
         mirrorShift(d);
      }
   }
\end{lstlisting}
Now, we make optimizations that are only possible in code. First, we utilize \hyperref[bp]{Theorem 1(iii)}, by storing only the first \(n\) letters of a bead, as it saves time and space. Next, using efficient letter manipulation and pointers within the strings, we save time and space when performing the mirror-shift operation, as we do not need to store the entire bead, nor do we use time to mirror its substring. 

Specifically, we only copy the portion of the bead that is to be mirrored, overwriting all other transformations on the existing structure. Rather than mirror the extracted substring, we assign a pointer to go from the end of the substring to the beginning, copying down letters to pre-determined indices.

On the following page, we provide an operational Java implementation of the program\footnote{Implementation Note: This program also returns the trivial superpermutations for $n=1,2$: $1$ and $121$, respectively.}.
\clearpage
\begin{lstlisting}[basicstyle={\scriptsize\ttfamily}]
public class SuperpermAlgo {

   private int n;
   private int intersect;
   private StringBuilder bead;
   private PrintWriter writer;

   public SuperpermAlgo(char[] c) {
      n = c.length;
      bead = new StringBuilder();
      try {writer = new PrintWriter("superperm_n" + n + ".txt");} 
      catch (IOException e) {e.printStackTrace();}
      startBead(c);
      algo(2);
      writer.close();
   }

   public void startBead(char[] c) {
      for (int i = 0; i < n; i++) {
         bead.append(c[i]);
      }
      for (int i = 0; i < bead.length(); i++) {
         writer.print(bead.charAt(i));
      }
      for (int i = 0; i < bead.length()-1; i++) {
         writer.print(bead.charAt(i));
      }
   }
   
   public void algo(int d) {
      if (d>=n) {return;}
      for (int i = 0; i < d; i++) {
         algo(d+1);
         intersect = d-1;
         if (i == d-1) {return;}
         mirrorShift(d);
      }
   }

   public void mirrorShift(int index) {
      int k = n-index;
      String m = bead.substring(0, k);
      for (int i = 0; i < n-k-1; i++) {
         bead.setCharAt(i, bead.charAt(i+k));
      }
      for (int i = n-k-1; i < n-1; i++) {
         if (m.length() == 0) {break;}
         int l = m.length()-1-(i-(n-k-1));
         if (l < 0) {continue;}
         bead.setCharAt(i, m.charAt(l));
      }
      for (int i = intersect; i < bead.length(); i++) {
         writer.print(bead.charAt(i));
      }
      for (int i = 0; i < bead.length()-1; i++) {
         writer.print(bead.charAt(i));
      }
   }
}
\end{lstlisting}
The time and space complexity for this program will be formally outlined in a subsequent section.\footnote{The .java files can be found in the supplementary materials.}

\section{Determining the Length of the Sequence} 

Now, we determine the length of the sequences generated using this algorithm. To do so, we add the length of all beads and then subtract all intersections. 

\subsection{The Initial Count}

Since there are \((n-1)!\) beads of length \(2n-1\) in the superpermutation for \(n\), the preliminary count will be 

\[(n-1)! \times (2n-1).\] 

\subsection{The Intersections}

\noindent
Now, we will account for the intersections. We know that the intersection length between two k-rings is \(\ell =n-k-2\). The following theorem will be utilized to determine the number of intersections between each order of ring.

\begin{theorem}[Number of Intersections per ring Order] \normalfont

In a superpermutation for a given \(n\), associated with k-rings, there are 
\[i(k) = (n-(k+1))!-1- \sum_{j = k+1}^{n-3} i(j)\]
intersections of the same length.
\end{theorem}

\begin{proof}
Firstly, we recall the nodes and edges argument from Theorem 2(ii). Using the same logic, we take the preliminary count of the number of intersections to be \((n-(k+1))!-1\). However, there is one key difference in this application; this function is double-counting all intersections between j-rings, such that \(k < j \le n-3\), since all k-rings are contained in these j-rings. 

Thus, the number of intersections between k-rings is recursively dependent on the sum of the number of intersections for all aforementioned j-rings, with \(n-3\)-rings being a base case, as there is precisely one intersection between the \(2\) rings of this order.
\end{proof}

\subsubsection{Re-indexing the Recurrence}
Corresponding to k-rings, the given intersection length is \(\ell = n-k-2\). Thus, the original index $k$ runs in the opposite direction to the lengths we care about: large values of $k$ (near $n-3$) give the shortest intersections, while small values of $k$ give the longest.

To align the index with the length itself, set \(j = n-k-2\). Then, \(j=1\) corresponds to \(k=n-3\) and intersections of length \(\ell=1\), and \(j=n-2\) corresponds to \(k=0\) and to the maximal intersection of length \(\ell = n-2\). Now, define 

\[
w_j = i(n-2-j),
\]

so that $w_j$ is both the number of intersections whose length equals $j$ and also indexed by that same length.

Substituting $k = n-2-j$ into Theorem 9 yields the forward‑running recurrence

\[
w_j = (j+1)! - 1 - \sum_{m=0}^{j-1} w_m ,
\quad 1 \le j \le n-2,
\]

with the understanding that $w_0 = 0$ (the `intersections' of length 0 contribute nothing). This index choice lets $j$ play a double role as both the loop variable and the intersection length, so later summations over $j$ automatically multiply the correct length by the correct frequency.

\subsubsection{Directly Defining the Recursive Relation}
We will now find a closed form equation in terms of \(n\) for the value of any index of \(w\). First, define the partial sum
\begin{center}
    $S_j = \displaystyle \sum_{k=0}^{j} w_k$.
\end{center}
Thus, we can see that
\begin{center}
    $S_j = S_{j-1} + w_j$.
\end{center}
From the original recursive relation, we know that
\begin{center}
    $w_j = (j+1)!-1-S_{j-1}$.
\end{center}
Next, we substitute this expression into the equation for $S_j$, yielding
\begin{center}
    $S_j = S_{j-1} + ((j+1)!-1-S_{j-1})$.
\end{center}
Simplifying, we obtain
\begin{center}
    $S_j = (j+1)!-1$.
\end{center}
Now, $w_j$ can be expressed as
\begin{center}
    $w_j=S_j-S_{j-1}$.
\end{center}
Inputting the calculated expressions for $S_j$ and $S_{j-1}$, we obtain
\begin{center}
    $w_j=((j+1)!-1)-(j!-1)$.
\end{center}
Simplifying yields the following:
\begin{center}
    $w_j=(j+1)! - j!$,
\end{center}
\begin{center}
    $w_j = (j+1)j!-j!$,
\end{center}
\begin{center}
    $w_j = ((j+1)-1)j!$,
\end{center}
and
\begin{center}
    $w_j = j \cdot j!$.
\end{center}
Thus, there are \(j \cdot j!\) intersections of length $j$ for \(1 \le j \le n-2\). Observe that we can simplify this relation into the following summation:
\[
\displaystyle\sum_{i=1}^{n-2}i(i \cdot i!).
\]
Simplifying yields:
\[
\displaystyle\sum_{i=1}^{n-2}i^2 \cdot i!.
\]
\subsection{Finalizing the Equation}
Thus, the final equation is
\[
l(n)=(2n-1)\times(n-1)!-\sum_{i=1}^{n-2}i^2 \cdot i!.
\]

\begin{remark}
    The closed-form length is equivalent to \(\sum_{k=1}^nk!\). This numerical match with the (now-disproven) minimal-length conjecture is coincidental; our aim is $\mathcal{O}(n)$ working space, not length optimality (see \citep{bib2} for the current best lower bound).
\end{remark}

\section{Time and Space Complexity Comparison of Methodologies}\label{tsc}
In this section, we provide a formal proof demonstrating that our mirror-shift algorithm is superior in terms of space efficiency compared to recursive and graph-theoretic approaches.

One assumption we will be making is that printing is equivalent to outputting in terms of time complexity, both being \(\mathcal{O}(n)\) for a string of length \(n\).

\subsection{Overview of Other Methods}

\subsubsection{Recursive Methods}

Recursive methods build superpermutations by \citep{bib1,bib4}:

\begin{itemize}
    \item Starting from a base case (e.g., \( k=2 \)) and incrementally constructing superpermutations for \( k+1 \) by inserting the new symbol between two copies of every permutation of \(k\), finally recombining all elements and eliminating overlap.
    \item Requiring the construction of superpermutations for all \( k < n \) before constructing for \( n \).
\end{itemize}

\subsubsection{Graph-Theoretic Methods}

Graph-theoretic approaches involve \citep{bib1,bib2,bib5}:

\begin{itemize}
    \item Modeling permutations as nodes in a graph, with edges representing overlaps between permutations.
    \item Seeking the shortest Hamiltonian path or walk that covers all permutations, which is an NP-hard problem.
\end{itemize}

\subsection{Time and Space Complexity Analysis}

\subsubsection{Mirror-shift Method}

\paragraph{Time Complexity}

\begin{theorem}[Time Complexity of the Mirror-shift Algorithm] \normalfont
The time complexity of the mirror-shift superpermutation construction method is \( \mathcal{O}(n!) \).
\end{theorem}

\begin{proof}
We proceed by analyzing the components of our method (assuming we are not utilizing mirroring)\footnote{Recall that the proof of the palindromic nature of the structure generated by our mirror-shift methodology enables a practical runtime speedup via an equivalent tradeoff in memory usage.}:
\begin{enumerate}
    \item \textbf{Initial Bead Generation}
    \begin{itemize}
        \item[(i)] Generating and printing the initial bead has a time complexity of \(\mathcal{O}(n)\).
    \end{itemize}
    \item \textbf{Mirror-shift Operations}
    \begin{itemize}
        \item[(i)] The number of mirror-shifts is \((n-1)!-1 = \mathcal{O}(n!).\)
        \item[(ii)] Each mirror-shift has an upper-bound time complexity of \(\mathcal{O}(n)\).
            \begin{itemize}
                \item[(a)] Copying the portion of the bead to be mirrored onto another string has the upper-bound time complexity of \(\mathcal{O}(n)\), since \(k \le n-3\) and we are mirroring \(k+1\) letters.
                \item[(b)] Finally, overwriting the letters with their new transformations in the current bead has an upper-bound time complexity of \(\mathcal{O}(n)\).\footnote{Rather than explicitly performing a mirroring operation, we can directly overwrite the bead with the new transformation by utilizing a pointer that goes from the end of the to-be-mirrored substring to the beginning rather than the other way around.}
            \end{itemize}
    \end{itemize}
    \item \textbf{Printing}
    \begin{itemize}
        \item[(i)] Printing each bead has an upper-bound time complexity of \(\mathcal{O}(n)\).
    \end{itemize}
\end{enumerate}

Therefore, the total time complexity is dominated by the mirror-shifts:
\[
T(n) = \mathcal{O}(n! \cdot n) = \mathcal{O}(n!).
\]
\end{proof}

\paragraph{Space Complexity}
\begin{theorem} \normalfont
The space complexity of the mirror-shift method is \(\mathcal{O}(n)\).
\end{theorem}

\begin{proof}
Our method optimizes space via incremental output. beads are printed as they are generated, avoiding the need to store the entire superpermutation in memory during its construction.

The upper-bound storage requirements are as follows:
\begin{itemize}
        \item[(i)] The first \(n\) letters of the current bead: \(\mathcal{O}(n)\).
        \item[(ii)] The recursive call stack: \(\mathcal{O}(n)\).
        \item[(iii)] The portion of the bead being mirrored in the mirror-shift function: \(\mathcal{O}(n)\).
        \item[(iv)] A small number of variables for indices and positioning in generation: \(\mathcal{O}(1)\).
\end{itemize}
Thus, the space complexity is asymptotically:
\[
\mathcal{S}(n) = \mathcal{O}(n).
\]
\end{proof}

\subsubsection{Recursive Methods}

\paragraph{Time Complexity}

\begin{theorem} \normalfont
The time complexity of recursive methods for constructing superpermutations is \(\mathcal{O}(n!)\).
\end{theorem}

\begin{proof}
Recursive methods require constructing superpermutations for all \(k\) from \(2\) up to \(n\):
\[
T(n) = \sum_{k=2}^{n} T(k)
\]
For each \(k\):
\[
T(k) = \mathcal{O}(k! \cdot k),
\]
since:
\begin{itemize}
    \item[(i)] There are \((k-1)!\) permutations to consider.
    \item[(ii)] Each permutation must be duplicated and appended to either end of \(k\), with the identification taking \(\mathcal{O}(k)\), concatenation is \(\mathcal{O}(k)\), and concatenating to the superpermutation being constructed is \(\mathcal{O}(k)\).
    \item[(iii)] Eliminating overlaps in the resulting sequence has a time complexity of \(\mathcal{O}(k)\) using the Knuth-Morris-Pratt pattern-finding algorithm.
    \item[(iv)] Outputting the resulting string requires \(\mathcal{O}(l)\) time, where \(l\) is the length of the string. A lower bound for this is \(k!\).
\end{itemize}
Therefore:
\[
T(n) = \sum_{k=2}^{n} k! = \mathcal{O}(n!)
\]
Hence, the asymptotic time complexity is \(\mathcal{O}(n!)\).
\end{proof}

\paragraph{Space Complexity}

\begin{theorem} \normalfont
The space complexity of recursive methods is \(\mathcal{O}(n!)\).
\end{theorem}

\begin{proof}
Recursive methods must store:

\begin{itemize}
    \item[(i)] The superpermutations for \(n\) and \(n-1\), which has a lower bound space complexity of \(\mathcal{O}(n!)\).
    \item[(ii)] The current permutation being used, which has a space complexity of \(\mathcal{O}(n)\).
    \item[(iii)] The recursive call stack, which can reach a depth of \(n\). 
    \item[(iv)] Auxiliary space for Knuth-Morris-Pratt and patternfinding: \(\mathcal{O}(n)\)
\end{itemize}
\noindent
Therefore, the total space complexity is dominated by the storage of the superpermutations:
\[
S(n) = \mathcal{O}(n!).
\]
\end{proof}

\subsubsection{Graph-Theoretic Methods}

\paragraph{Time Complexity}

\begin{theorem} \normalfont
The time complexity of graph-theoretic methods is \(\mathcal{O}(2^{n!}n!^2)\) or \(\mathcal{O}((n!)!)\).
\end{theorem}

\begin{proof}
The components of Graph-theoretic methods are as follows:

\begin{itemize}
    \item Graph Construction:
    \begin{itemize}
        \item[(i)] Generating nodes, with each representing a unique permutation: \(\mathcal{O}(n \cdot n!)\).
        \item[(ii)] Generating an edge between each pair of nodes, with each edge representing the overlap between permutations (Knuth-Morris-Pratt): \(\mathcal{O}(\binom{n!}{2} \cdot n)\).
    \end{itemize}
    \item Hamiltonian Path Search (Two Approaches):
    \begin{itemize}
        \item[(i)] Brute-forcing the shortest path has a time complexity of 
        \(\mathcal{O}(k!)\) for a graph containing \(k\) nodes: \(\mathcal{O}((n!)!)\).
        \item[(ii)] Using dynamic programming, approaches such as Held-Karp still require exponential time complexity; \(\mathcal{O}(2^k \cdot k^2)\) for a graph containing \(k\) nodes: \(\mathcal{O}(2^{n!}n!^2)\).
    \end{itemize}
\end{itemize}
\noindent
Therefore, the time complexity is dominated by the evaluation of the shortest Hamiltonian Path, with a best-case of
\[
T(n) = \mathcal{O}(2^{n!}n!^2)
\]
using Held-Karp, and a worst-case time complexity of
\[
T(n) = \mathcal{O}((n!)!)
\]
using brute force.

\end{proof}

\paragraph{Space Complexity}

\begin{theorem} \normalfont
The space complexity of graph-theoretic methods is \( \mathcal{O}((n!)^2) \).
\end{theorem}

\begin{proof}
Graph-theoretic approaches must store:

\begin{itemize}
        \item[(i)] The nodes of the graph, where each represents one \(n\)-length permutation; for \(n!\) permutations: \( \mathcal{O}(n! \cdot n) \).
        \item[(ii)] The edges of the graph, where there are \(\binom{k}{2}\) for \(k\) nodes; for \(n!\) nodes: \(\mathcal{O}(\binom{n!}{2})\).
        \item[(iii)] Auxiliary space for pathfinding algorithms; \(\mathcal{O}(n \cdot 2^n)\) for Held-Karp, or \(\mathcal{O}(n)\) for brute force.
\end{itemize}
\noindent
Thus, the space complexity is dominated by the storage of the edges of the graph:\footnote{Note that space optimization methods such as dynamic edge calculation can reduce the asymptotic space complexity to \(\mathcal{O}(n!)\) with an equal tradeoff in time complexity. In this case, \(\mathcal{O}(\binom{n!}{2} \cdot n)\) would be added to computation time.}
\[
S(n) = \mathcal{O}(n!^2).
\]
\end{proof}

\subsection{Comparative Analysis}

\subsubsection{Time Complexity Analysis}

Our methodology achieves minimal time complexity for constructing a superpermutation containing all \(n!\) permutations. In practical applications, recursive methods have additional overhead due to cumulative construction for all \(k < n\), and graph-theoretic methods suffer from the NP-hardness of finding the shortest Hamiltonian path.

\subsubsection{Space Complexity Analysis}

Our mirror-shift algorithm is highly space-efficient, requiring the storage proportional to half of a single bead (length \(n\)), the recursive stack, a small substring, and two variables. In contrast, recursive methods must store the superpermutations for \(n\) and \(n-1\), and graph-theoretic methods require immense space to store the entire graph structure.
\clearpage

\input{sn-article.bbl}

\end{document}